\documentclass[11pt]{article} 
\usepackage{amsfonts} 
\usepackage{psfrag}
\def\bbz{{\mathbb Z}} 

\newtheorem{theorem}{Theorem}[section]
\newtheorem{conjecture}[theorem]{Conjecture} 
 \newtheorem{lemma}[theorem]{Lemma}

\def\qed{ \ \vrule width.2cm height.2cm depth0cm\smallskip}
\newenvironment{proof}{\noindent{\bf Proof.}}{\qed}

\begin{document}

\title{{\Large {\bf Edge-Coloring Series-Parallel Multigraphs}}}

\author{
\parbox{6cm} { 
Cristina G. Fernandes 
\thanks{Research supported in part by CNPq Proc.~No.~301174/97-0, 
FAPESP Proc.~No.~98/14329 and~96/04505-2 and 
PRONEX/CNPq 664107/1997-4 (Brazil).} \\ 
{\footnotesize Departamento de Ci\^encia da Computa\c c\~ao \\ 
  Instituto de Matem\'atica e Estat\'{\i}stica \\ 
  Universidade de S\~ao Paulo - Brazil \\ E-mail: cris@ime.usp.br}
} \ \ \ and \ \ \
\parbox{6cm} { 
Robin Thomas 
\thanks{Research supported in part by NSA under Grant
  No.~MDA904-98-1-0517 and by NSF under Grant No.~DMS-9970514.} \\ 
{\footnotesize School of Mathematics \\ 
  Georgia Institute of Technology \\ 
  Atlanta, GA 30332-0160, USA \\ E-mail: thomas@math.gatech.edu} }}

\date{\today}

\maketitle

\begin{abstract}
We give a simpler proof of Seymour's Theorem on edge-coloring series-parallel
multigraphs and derive a linear-time algorithm to check whether a given
series-parallel multigraph can be colored with a given number of colors.
\end{abstract}

\section{Introduction}

All {\em graphs} in this paper are finite, may have parallel edges,
but no loops. Let $k\ge 0$ be an integer.  A graph $G$ is {\em
$k$-edge-colorable}\ if there exists a map $\kappa: E(G) \rightarrow
\{1, \ldots, k\}$, called a {\em $k$-edge-coloring}, such that
$\kappa(e) \neq \kappa(f)$ for any two distinct edges $e, f$ of $G$
that share at least one end. The {\em chromatic index $\chi'(G)$}\ is
the minimum $k \geq 0$ such that $G$ is $k$-edge-colorable.  Clearly
$\chi'(G) \geq \Delta(G)$, where $\Delta(G)$ is the maximum degree of
$G$, but there is another lower bound.  Let $$ \Gamma(G) = \max\left\{
\frac{2|E(G[U])|}{|U|-1}: U \subseteq V(G), |U| \geq 3 \mbox{\ and
$|U|$ is odd}\right\}.$$ If $U$ is as above, then every matching in
$G[U]$, the subgraph induced by $U$, has size at most
$\lfloor{1\over2}|U|\rfloor$. Consequently, $\chi'(G) \geq \Gamma(G)$.
If $G$ is the Petersen graph, or the Petersen graph with one vertex
deleted, then $\chi'(G)>\max\{\Delta(G), \lceil\Gamma(G)\rceil\}$.
However, Seymour conjectures that equality holds for planar graphs:

\begin{conjecture}
\label{seymourconj}
If $G$ is a planar graph, then
$\chi'(G)=\max\{\Delta(G),\lceil\Gamma(G)\rceil\}$.
\end{conjecture}

Conjecture~\ref{seymourconj} most likely does not have an easy proof,
because it implies the Four-Color Theorem. 
Marcotte~\cite{M01} proved that this conjecture holds for graphs 
which do not contain $K_{3,3}$ and do not contain $K_5 \setminus e$ 
as a minor (where $K_5 \setminus e$ is the graph obtained from $K_5$ 
by removing one of its edges). 
This result extended a previous result by Seymour~\cite{S90}, who
proved that his conjecture holds for series-parallel graphs (a graph
is {\em series-parallel}\ if it has no subgraph isomorphic to a
subdivision of $K_4$):

\begin{theorem}
\label{seymourthm}
If $G$ is a series-parallel graph, and $k$ is an integer with $k \geq
\max\{\Delta(G), \Gamma(G)\}$ then $G$ is $k$-edge-colorable.
\end{theorem}

It should be noted that Theorem~\ref{seymourthm} is fairly easy for simple
graphs; the difficulty lies in the presence of parallel edges.  Seymour's
proof is elegant and interesting, but the induction step requires the
verification of a large number of inequalities. We give a simpler proof, based
on a structural lemma about series-parallel graphs, which in turn is an easy
consequence of the well-known fact that every simple series-parallel graph has
a vertex of degree at most two. Our work was motivated by the list
edge-coloring conjecture of \cite{BolHar} (see also \cite[Problem 12.20]{JT}):

\begin{conjecture}
\label{listconj}
Every graph is $\chi'(G)$-edge-choosable.
\end{conjecture}

At present there seems to be no credible approach for proving the conjecture in
full generality. We were trying to gain some insight by studying it for
series-parallel graphs. The conjecture has been verified for {\em simple}
series-parallel graphs in~\cite{JMT99}, but it is open for series-parallel
graphs with parallel edges. Our efforts only resulted in a simpler proof of
Theorem~\ref{seymourthm} and in a linear-time algorithm for checking whether or
not a series-parallel graph can be colored with a given number of colors.
Our algorithm substantially simplifies an earlier algorithm of Zhou, Suzuki 
and Nishizeki~\cite{ZhouSN96}.

\section{Three lemmas}

For our proof of Theorem~\ref{seymourthm} we need three lemmas.  The
first two are easy, and the third appeared in \cite{JMT99}.  Let $G$
be a graph, and let $u,v$ be adjacent vertices of $G$. We use $uv$ to
denote the unique edge with ends $u$ and $v$ in the underlying simple
graph of $G$. If $G$ has $m$ edges with ends $u$ and~$v$, then we say
that $uv$ has {\em multiplicity} $m$. If $u$ and $v$ are not adjacent,
then we say that $uv$ has multiplicity zero.  Let $G$ be a graph, let
$\kappa$ be a $k$-edge-coloring of a subgraph $H$ of $G$, let $u\in
V(G)$, and let $i\in\{1,2,\ldots,k\}$.  We say that $u$ {\em sees} $i$
and that $i$ {\em is seen by} $u$ if $\kappa(f)=i$ for some edge $f$
of~$H$ incident with $u$.

\begin{lemma}
\label{colorlemma}
Let $G$ be a graph, let $u_0\in V(G)$, let $u_1,u_2$ be distinct neighbors
of $u_0$, let $H$ be the graph obtained from $G$ by deleting all edges
with one end $u_0$ and the other end $u_1$ or $u_2$, and let
$\kappa$ be a $k$-edge-coloring of $H$. For $i=1,2$ let $m_i$ be
the multiplicity of $u_0u_i$ in $G$, and 
for $i=0,1,2$ let $S_i$ be the set of colors seen by $u_i$.  
If $m_1+|S_0\cup S_1|\le k$, $m_2+|S_0\cup S_2|\le k$ and
 $m_1+m_2 +|S_0\cup (S_1\cap S_2)|\le k$, 
then $\kappa$ can be extended to a $k$-edge-coloring of $G$.
\end{lemma}

\begin{proof}
  Since $m_1+|S_0\cup S_1|\le k$, the edges with ends $u_0$ and $u_1$
  can be colored using colors not in $S_0\cup S_1$.  We do that, using
  as many colors in $S_2$ as possible.  If the $u_0 u_1$ edges can be
  colored using colors in $S_2$ only, then there are at least
  $k-|S_0\cup S_2|\ge m_2$ colors left to color the edges with ends
  $u_0$ and $u_2$, and so $\kappa$ can be extended to a
  $k$-edge-coloring of $G$, as desired. Otherwise, the $u_0u_1$ edges
  of~$G$ will be colored using $|S_2 \setminus (S_0\cup S_1)|$ colors
  from $S_2$, and $m_1-|S_2 \setminus (S_0\cup S_1)|$ other colors.
  Thus the number of colors available to color the $u_0u_2$ edges of
  $G$ is at least $k-|S_0\cup S_2|-(m_1-|S_2 \setminus (S_0\cup
  S_1)|)=k-m_1-|S_0\cup (S_1\cap S_2)|\ge m_2$, and so the coloring
  can be completed to a $k$-edge-coloring of $G$, as desired.
\end{proof}

%

\begin{lemma}
\label{ineq}
Let $k$ be an integer, and let $G$ be a graph with $\Delta(G)\le k$.
Then $\Gamma(G)\le k$ if and only if $2|E(G[U])|\le k(|U|-1)$ for
every set $U \subseteq V(G)$ such that $|U|$ is odd and at least
three, and the underlying graph of $G[U]$ has no vertices of degree at
most one. 
\end{lemma}

\begin{proof} The ``only if'' part is clear. To prove the ``if'' part
  we must show that $2|E(G[U])|\le k(|U|-1)$ for every set $U
  \subseteq V(G)$ such that $|U|$ is odd and at least three.  We
  proceed by induction on $|U|$. We may assume that the underlying
  graph of $G[U]$ has a vertex $u$ of degree at most one, for
  otherwise the conclusion follows from the hypothesis. If $u$ has
  degree one in the underlying graph of $G[U]$, then let $v$ be its
  unique neighbor; otherwise let $v\in U\backslash \{u\}$ be
  arbitrary.
  Let $U'=U\backslash \{u,v\}$. Then $2|E(G[U])|\le 2\Delta(G) +
  2|E(G[U'])|\le 2k+k(|U'|-1)\le k(|U|-1)$ by the induction hypothesis
  if $|U|>3$ and trivially otherwise, as desired.
\end{proof}

The third lemma appeared in~\cite{JMT99}. For the sake of completeness
we include its short proof.

\begin{lemma} \label{l:main}
Every non-null simple series-parallel graph $G$ has one of the
following:
\vspace*{-2mm}
\begin{itemize}
\item[(a)] a vertex of degree at most one,
\vspace*{-2mm}
\item[(b)] two distinct vertices of degree two with the same
neighbors,
\vspace*{-2mm}
\item[(c)] two distinct vertices $u,v$ and two not necessarily
distinct vertices $w,z\in V(G)\backslash\{u,v\}$ such that the
neighbors of $v$ are $u$ and $w$, and every neighbor of $u$ is equal
to $v$, $w$, or $z$, or
\vspace*{-2mm}
\item[(d)] five distinct vertices $v_1,v_2,u_1,u_2,w$ such that the
neighbors of $w$ are $u_1,u_2,v_1,v_2$, and for $i=1,2$ the neighbors
of $v_i$ are $w$ and $u_i$.
\end{itemize}
\end{lemma}

\begin{proof}
We proceed by induction on the number of vertices. Let $G$ be a
non-null simple series-parallel graph, and assume that the result
holds for all graphs on fewer vertices. We may assume that $G$ does
not satisfy (a), (b), or (c).  Thus $G$ has no two adjacent vertices
of degree two.  By suppressing all vertices of degree two (that is,
contracting one of the incident edges) we obtain a series-parallel
graph without vertices of degree two or less. Therefore, by a
well-known property of series-parallel graphs~\cite{Du}, this graph is
not simple.  Since $G$ does not satisfy (b), this implies that~$G$ has
a vertex of degree two that belongs to a cycle of length three.  Let
$G'$ be obtained from $G$ by deleting all vertices of degree two that
belong to a cycle of length three. First notice that if $G'$ has a
vertex of degree less than two, then the result holds for $G$ (cases
(a), (b), or case (c) with $w=z$). Similarly, if $G'$ has a vertex of
degree two that does not have degree two in $G$, then the result holds
(one of the cases (b)--(d) occurs).  Thus we may assume that $G'$ has
minimum degree at least two, and every vertex of degree two in $G'$
has degree two in $G$. By induction, (b), (c), or~(d) holds for~$G'$,
but it is easy to see that then one of (b), (c), or (d) holds for $G$.
\end{proof}



\section{Proof of Theorem~\ref{seymourthm}}



We proceed by induction on $|E(G)|$, and, subject to that, by
induction on $|V(G)|$. The theorem clearly holds for graphs with no
edges, so we assume that $G$ has at least one edge, and that the
theorem holds for graphs with fewer edges or the same number of edges
but fewer vertices.  Let $S$ be the underlying simple graph of $G$.
We apply Lemma \ref{l:main} to $S$, and distinguish the corresponding
cases.

If case (a) holds, let $G'$ be the graph obtained from $G$ by removing
a vertex of degree at most one in $S$. The rest is straightforward: $k
\geq \max\{\Delta(G'), \Gamma(G')\}$ and so, by induction, there is a
$k$-edge-coloring of $G'$. From this $k$-edge-coloring, it is easy to
obtain a $k$-edge-coloring for $G$.

\begin{figure}[b]
\psfrag{(a)}{(a)}
\psfrag{(b)}{(b)}
\psfrag{G}{$G$}
\psfrag{G'}{$G'$}
\psfrag{u}{$u$}
\psfrag{v}{$v$}
\psfrag{x}{$x$}
\psfrag{y}{$y$}
\psfrag{a}{$a$}
\psfrag{b}{$b$}
\psfrag{c}{$c$}
\psfrag{d}{$d$}
\psfrag{a-d}{$a-d$}
\psfrag{b+d}{$b+d$}
\begin{center}\leavevmode%
\includegraphics{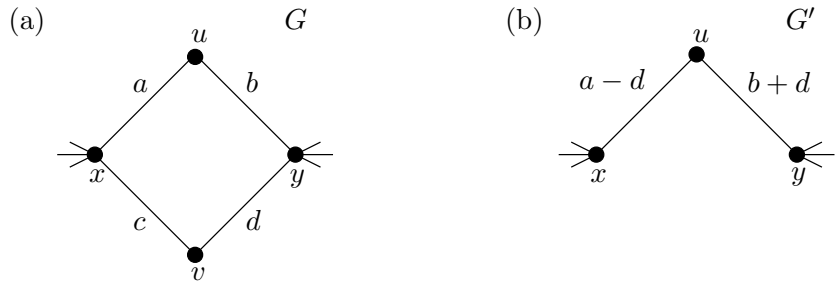}
\end{center}
\caption{Configurations referring to Case (b).}
\label{CaseB}
\end{figure}

If case (b) holds, let $u$ and $v$ be two distinct vertices of degree two in $S$
with the same neighbors. Let the common neighbors be $x$ and $y$.  Let $a,b,c,d$
be the multiplicities of $ux,uy,vx,vy$, respectively.  See
Figure~\ref{CaseB}(a). From the symmetry we may assume that $a\ge d$.  Let $G'$
be obtained from $G\backslash v$ by deleting $d$ edges with ends $u$ and $x$,
and adding $d$ edges with ends $u$ and $y$. See Figure~\ref{CaseB}(b). Then
clearly $\Delta(G')\le k$, and it follows from Lemma~\ref{ineq} that
$\Gamma(G')\le k$. By the induction hypothesis the graph $G'$ has a
$k$-edge-coloring $\kappa'$. Let $A$ be a set of colors of size $d$ used by a
subset of the edges of $G'$ with ends $u$ and $y$, chosen so that as few as
possible of these colors are seen by $x$.  By deleting those edges we obtain a
coloring of $G\backslash v$, where $d$ edges with ends $u$ and $x$ are
uncolored. Next we color those $d$ uncolored edges, first using colors in $A$
not seen by $x$, and then using arbitrary colors not seen by $x$ or $u$. This
can be done: if at least one color in $A$ is seen by $x$, then once we exhaust
colors of $A$ not seen by $x$, the choice of $A$ implies that every color seen
by $u$ is seen by $x$, and so the coloring can be completed, because $x$ has
degree at most $k$. This results in a $k$-edge-coloring of $G\backslash v$ with
the property that at least $d$ of the colors seen by $x$ (namely the colors in
$A$) are not seen by $y$. Thus the number of colors seen by both $x$ and $y$ is
at most $k-c-d$ ($v$ sees no colors), and clearly the number of colors seen by
$x$ is at most $k-c$ and the number of colors seen by $y$ is at most $k-d$. By
Lemma~\ref{colorlemma} this coloring can be extended to a $k$-edge-coloring of
$G$, as desired.

We now assume a special case of (c) of Lemma~\ref{l:main}. Let $u, v, w, z$ be
as in that lemma, with $w=z$. Then clearly $\Delta(G\backslash v)\le k$ and
$\Gamma(G\backslash v)\le k$, and so $G\backslash v$ has a
$k$-edge-coloring. This $k$-edge-coloring can be extended to a $k$-edge-coloring
of $G$ by first coloring the edges with ends $w$ and $v$ (this can be done
because the degree of $w$ is at most $k$), and then coloring the edges with ends
$u$ and $v$ (there are enough colors for this because $|E(G[U])| \le k$ for $U =
\{u,v,w\}$).

Finally we assume that case (d) of Lemma~\ref{l:main} holds and we will show
that our analysis includes the remainder of case (c) as a special case.  Let
$v_1, v_2, u_1, u_2$ and $w$ be as in the statement of Lemma~\ref{l:main}, and
let $a$, $b$, $c$, $d$, $e$ and $f$ be the multiplicities of $u_1v_1$, $u_1w$,
$v_1w$, $v_2w$, $u_2w$ and $u_2v_2$, respectively, as in Figure~\ref{CaseD}(a).  
In order to include case (c) we will not be assuming that $a$, $b$, $c$, $d$, 
$e$ and $f$ are nonzero; we only assume that $c+d>0$. (This is why the primary
induction is on $|E(G)|$.)  If $a+b+c+d+e+f\le k$, then a $k$-edge-coloring of
$G\backslash \{v_1w,v_2w\}$ can be extended to a $k$-edge-coloring of
$G$, and so we may assume that $k<a+b+c+d+e+f$. Since $w$ has degree
at most $k$ we have $b+c+d+e\le k$, and by considering the sets
$U=\{u_1,v_1,w\}$ and $U=\{u_2,v_2,w\}$ we deduce that $a+b+c\le k$
and $d+e+f\le k$. Let $z_1 = \max\{0, a+b+c+e-k\}$, $z_2 = \max\{0,
b+d+e+f-k\}$ and $s = k-(b+c+d+e)$. Thus $z_1\le e$, $z_2\le b$,
$s\ge0$ and 
\begin{eqnarray}
a+f-z_1-z_2-s=\cases{k-(b+e) & if $z_1>0$ and $z_2>0$\cr
  a+c & if $z_1=0$ and $z_2>0$\cr d+f & if $z_1>0$ and $z_2=0$\cr
  a+f-s & if $z_1=z_2=0$.}
\label{equa}
\end{eqnarray}
We claim that there exist nonnegative integers $s_1$ and $s_2$ such
that $s = s_1+s_2$, $s_1 \leq a-z_1$ and $s_2 \leq f-z_2$.  To prove
this claim it suffices to check that $a-z_1\ge0$, $f-z_2\ge0$ and
$a-z_1+f-z_2\ge s$. We have $a-z_1 \geq \min\{a,k-(b+c+e)\}
\geq\min\{a,d\} \ge 0$, and by symmetry $f-z_2 \ge 0$.  The third
inequality follows from (\ref{equa}).  This proves the existence of 
$s_1$ and $s_2$. 

\begin{figure}[ht]
\psfrag{(a)}{(a)}
\psfrag{(b)}{(b)}
\psfrag{u1}{$u_1$}
\psfrag{u2}{$u_2$}
\psfrag{v1}{$v_1$}
\psfrag{v2}{$v_2$}
\psfrag{x}{$x$}
\psfrag{y}{$y$}
\psfrag{w}{$w$}
\psfrag{a}{$a$}
\psfrag{b}{$b$}
\psfrag{c}{$c$}
\psfrag{d}{$d$}
\psfrag{e}{$e$}
\psfrag{f}{$f$}
\psfrag{a-z1-s1}{$a{-}z_1{-}s_1$}
\psfrag{f-z2-s2}{$f{-}z_2{-}s_2$}
\psfrag{z1+z2}{$z_1{+}z_2$}
\psfrag{b-z2}{$b{-}z_2$}
\psfrag{e-z1}{$e{-}z_1$}
\begin{center}\leavevmode%
\includegraphics{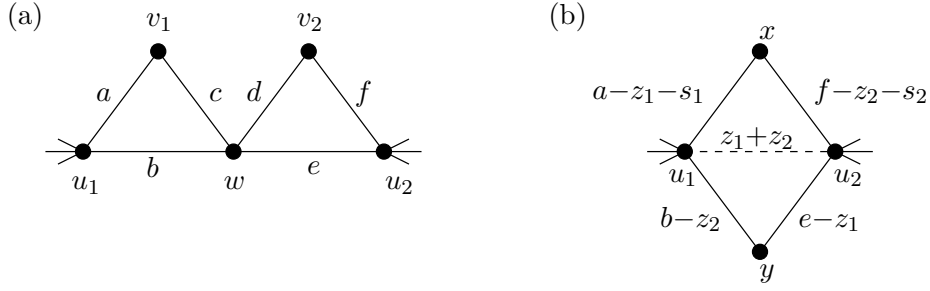}
\end{center}
\caption{Configurations referring to Case (d).}
\label{CaseD}
\end{figure}

Let $G'$ be obtained from $G$ by removing the vertices $v_1$, $v_2$, $w$,
adding two new vertices, $x$ and~$y$, and adding $a-z_1-s_1$ edges with ends
$x$ and $u_1$, $f-z_2-s_2$ edges with ends $x$ and $u_2$, $b-z_2$ edges with
ends $y$ and $u_1$, $e-z_1$ edges with ends $y$ and $u_2$, and $z_1+z_2$ edges
with ends $u_1$ and $u_2$. See Figure~\ref{CaseD}(b). Thus $|E(G')|<|E(G)|$.

It follows from (\ref{equa}) that $x$ has degree at most $k$.  Since all
other vertices of $G'$ clearly have degree at most $k$, we see that $k
\geq \Delta(G')$. We claim that $k \geq \Gamma(G')$. By
Lemma~\ref{ineq} we must show that $2|E(G'[X'])| \leq k(|X'|-1)$ for
every set $X'\subseteq V(G')$ such that $|X'|$ is odd, $|X'|\ge3$ and
the underlying graph of 
$G'[X']$ has no vertices of degree at most one. 
If $|X'\cap\{u_1,u_2\}|\le1$, then $G[X']=G'[X']$, and the result
follows. Thus we may assume that $u_1,u_2\in X'$. We need to
distinguish several cases. 
If $x,y\in X'$, then let $X = X' \setminus \{x,y\}$. We have
$2|E(G'[X'])|=2|E(G[X])|+2(a-z_1-s_1+f-z_2-s_2+z_1+z_2+b-z_2+e-z_1)\le
k(|X'|-1)$, using the induction hypothesis and the relations
$s_1+s_2=k-(b+c+d+e)$, $z_1\ge a+b+c+e-k$ and $z_2\ge b+d+e+f-k$.  If
$x\in X'$ and $y\not\in X'$ we put $X = X' \setminus \{x\} \cup
\{w,v_1,v_2\}$, and if $x\not\in X'$ and $y\in X'$ we put $X = X'
\setminus \{y\} \cup \{w\}$. In either of these two cases the counting
is straightforward. Finally, we assume that $x,y\not\in X'$. If
$z_1=z_2=0$, then $G[X']=G'[X']$, and so the conclusion holds. If
$z_1>0$ and $z_2>0$, then let $X = X' \setminus \{u_1,u_2\}$. We have
$2|E(G'[X'])|\le 2|E(G[X])|+2(k-(a+b)+k-(e+f)+z_1+z_2)\le
k(|X|-1)+2(b+c+d+e)\le k(|X'|-1)$, where the second inequality follows
from the induction hypothesis (or is trivial if $|X|=1$) and the
definition of $z_1$ and $z_2$. Finally, from the symmetry between
$z_1$ and $z_2$ it suffices to consider the case $z_1=0$ and
$z_2>0$. In that case we put $X = X' \cup \{w,v_2\}$. Then
$2|E(G'[X'])|=2|E(G[X])|+2(z_1+z_2-(b+d+e+f))\le k(|X'|-1)$, using the
induction hypothesis and the definition of $z_1$ and $z_2$.  This
completes the proof that $k \geq \Gamma(G')$.

By induction there exists a $k$-edge-coloring $\kappa'$ of $G'$.  Let
$Z_1\cup Z_2$ be the colors used on the $z_1+z_2$ edges of
$E(G') \setminus E(G)$ with ends $u_1$ and $u_2$, so that $|Z_1|=z_1$ and
$|Z_2|=z_2$.  Let $G''$ be the graph obtained from $G$ by deleting all
edges with one end $w$ and the other end $v_1$ or $v_2$. We first
construct a suitable $k$-edge-coloring $\kappa''$ of $G''$. To do so
we start with the restriction of $\kappa'$ to $E(G'')\cap E(G')$, and
then use $Z_1$ and the colors of the $xu_1$ edges of $G'$ to color a
subset of the $u_1v_1$ edges of $G$, we use $Z_2$ and the colors of
the $yu_1$ edges of $G'$ to color all of the $wu_1$ edges of $G$, and
symmetrically we use $Z_1$ and the colors of the $u_2y$ edges of $G'$
to color all the $wu_2$ edges of $G$, and we use $Z_2$ and all the
colors of the $xu_2$ edges of $G'$ to color a subset of the $v_2u_2$
edges of $G$. We color the $s_1$ uncolored $u_1v_1$ edges and the
$s_2$ uncolored $u_2v_2$ edges arbitrarily. That can be done, because
$u_i$ is the only neighbor of $v_i$ in $G''$.  This completes the
definition of $\kappa''$.  Now the number of colors seen by $v_1$ or
$w$ is at most $a-z_1-s_1+z_1+z_2+b-z_2+e-z_1+s_1=a+b+e-z_1\le k-c$,
and similarly the number of colors seen by $v_2$ or $w$ is at most
$k-d$.  The number of colors seen by $w$, or by both $v_1$ and $v_2$
is at most $b-z_2+e-z_1+z_1+z_2+s\le k-(c+d)$.  By
Lemma~\ref{colorlemma} the $k$-edge-coloring $\kappa''$ can be
extended to a $k$-edge-coloring of $G$, as desired.


\section{A linear-time algorithm}\label{algosec}
In this section we present a linear-time algorithm to decide whether
$\chi'(G)\le k$, where the series-parallel graph $G$ and the
integer $k$ are part of the input instance.  The idea of the algorithm
is very simple -- we repeatedly find vertices of the underlying simple
graph satisfying one of (a)--(d) of Lemma~\ref{l:main}, construct the
graph $G'$ as in the proof of Theorem~\ref{seymourthm}, apply the
algorithm recursively to $G'$ to check whether $\chi'(G')\le k$, and
from that knowledge we deduce whether $\chi'(G)\le k$.  
The construction of $G'$ is straightforward, and the decision whether
$\chi'(G)\le k$ is easy: suppose, for instance, that we find vertices
$v_1,v_2,u_1,u_2,w$ as in Lemma~\ref{l:main}(d), and let $a,b,c,d,e,f$
be as in the proof of Theorem~\ref{seymourthm}. If $a+b+c+d+e+f\ge k$,
then construct $G'$ as in the proof; we have $\chi'(G)\le k$ if and only if
$\chi'(G')\le k$ and $a+b+c\le k$ and $d+e+f\le k$. If $a+b+c+d+e+f\le
k$, then $\chi'(G)\le k$ if and only if $\chi'(G\backslash w)\le k$.
Thus it remains to describe how to find the vertices as in
Lemma~\ref{l:main}.  That can be done by a slight modification of a
linear-time recognition algorithm for series-parallel graphs.  We need
a few definitions in order to describe the algorithm.

Let $H$ be a graph, and let $\lambda$ be a function assigning to each edge $e
\in E(H)$ a set $\lambda(e)$ disjoint from $V(H)$ in such a way that $\lambda(e)
\cap \lambda(e') = \emptyset$ for distinct edges $e, e'\in E(H)$. Let
$H_\lambda$ be the graph obtained from $H$ by adding, for each edge $e\in E(H)$
and each $x \in \lambda(e)$, a vertex $x$ of degree two, adjacent to the two
ends of $e$.  Then $H_\lambda$ is unique up to isomorphism, and so we can speak
of the graph $H_\lambda$.  Now let $\mu:E(H_\lambda)\to\bbz^+_0$ be a function,
and let $H^\mu_\lambda$ be the graph obtained from $H_\lambda$ by replacing each
edge $e\in E(H_\lambda)$ by $\mu (e)$ parallel edges with the same ends. In
those circumstances we say that $(H,\lambda,\mu)$ is an {\em encoding}, and that
it is an {\em encoding} of $H^\mu_\lambda$.

For a graph $H$ and $v\in V(H)$ we let deg$_H(v)$ denote the number of
edges incident to $v$ in $H$ and val$_H(v)$ denote the number of
distinct neighbors of $v$ in $H$. Thus val$_H(v)\le\deg_H(v)$ with
equality if and only if $v$ is incident with no parallel edges. We say
that a function $C:V(H) \to\bbz^+_0$ is a {\em counter} for a graph
$H$ if $\deg_H(v)-\mbox{val}_H(v)\le C(v)$ for every vertex $v\in
V(H)$.  We say that a vertex $v\in V(H)$ is {\em active} if either
$\deg_H(v)\le 2$ or $\deg_H(v)\le 3C(v)$.

The following lemma guarantees that if there are no active vertices,
then the graph is null.

\begin{lemma}\label{active}
Let $H$ be a non-null series-parallel graph, and let $C$ be a counter for $H$.
Then there exists an active vertex.
\end{lemma}

\begin{proof}
As noted in the proof of Lemma~\ref{l:main}, the underlying simple graph of
$H$ has a vertex of degree at most two.  Thus $H$ has a vertex $v$ with
val$_H(v)\le 2$.  If $\deg_H(v)> 3C(v)$, then
$$\deg(v)-2\le \deg_H(v)-\mbox{val}_H(v)\le C(v)< \deg_H(v)/3,$$
which implies $\deg_H(v)\le 2$.  Thus $v$ is active, as desired.
\end{proof}

\subsection{The algorithm}

The input for the algorithm is a series-parallel graph $G$ and a non-negative
integer $k$, where the graph $G$ is presented by means of its underlying
undirected graph and a function $E(G) \to\bbz^+$ that describes the
multiplicity of each edge. 

The algorithm starts by checking whether $\deg_G(v) \le k$ for all $v \in
V(G)$. If not, it outputs ``no, $\chi'(G)\not\le k$'' and terminates. Otherwise
let $H$ be the underlying undirected graph of $G$, let $\lambda(e) := \emptyset$
for every edge $e \in E(H)$, let $\mu(e)$ be the multiplicity of $e$ in $G$, and
let $C(v) := 0$ for every $v \in V(H)$.  Then $(H,\lambda,\mu)$ is an encoding
of $G$ and $C$ is a counter for $H$. The algorithm computes the list of all
active vertices of $H$. It does not matter how $L$ is implemented as long as
elements can be deleted and added in constant time.

After this, the algorithm is iterative. Each iteration starts with an encoding
$(H,\lambda,\mu)$ of the current series-parallel graph $G$, a counter $C$ for
$H$ and a list $L$ which includes all active vertices of~$H$.

Each iteration consists of the following. If $L = \emptyset$, then we output
``yes, $\chi'(G)\le k$'' and terminate, else we let $v$ be a vertex in $L$. If
$v \not\in V(H)$ or $v$ is not active, then we remove $v$ from $L$ and move to
the next iteration. If $v \in V(H)$ and $v$ is active, then there are three
possible cases.

If $\deg_H(v) > 2$, then $\deg_H(v) \le 3C(v)$, because $v$ is active. We
rearrange the adjacency list of~$v$, removing all but one edge from each class
of parallel edges incident with $v$, adjusting $\lambda$ and $\mu$ so that
$(H,\lambda,\mu)$ is still an encoding of $G$. We set $C(v) := 0$, include in
$L$ all vertices whose degree decreased and move to the next iteration.

If $\deg_H(v) = \mbox{val}_H(v) = 2$ and $\lambda(vx) = \lambda(vy) = 0$, where
$x$ and $y$ are the two distinct neighbors of $v$, then we remove $v$ from $H$
and add a new edge $f=xy$ to $H$. We set $\mu(f):=0, \lambda(f):=\{v\}$,
increase both $C(x)$ and $C(y)$ by one, add $x$ and $y$ to $L$ and move to the
next iteration.

If $\deg_H(v) \le 2$ but the previous case does not apply, then we have located
vertices of $G$ satisfying one of (a) to (d) of Lemma~\ref{l:main}. We check if
the local conditions are satisfied or not (for example, in case (d), if
$a+b+c+d+e+f\ge k$, we check whether $a+b+c\le k$ and $d+e+f\le k$); if they are
not, we output ``no, $\chi'(G)\not\le k$'' and terminate.  Otherwise, we modify
the encoding $(H,\lambda,\mu)$ to get an encoding of the graph $G'$ described in
the proof of Theorem~\ref{seymourthm}. This involves deleting vertices from $H$
and adding edges to $H$. Every time an edge of $H$ incident with a vertex $z\in
V(H)$ is deleted or added we increase $C(z)$ by one and add $z$ to $L$. We move
to the next iteration.

\vspace{3mm}

The correctness of the algorithm follows from Lemma~\ref{active} and from the
proof of Theorem~\ref{seymourthm}. 

To analyze the running-time, let $n$ denote the number of vertices of the input
graph $G$. The initial steps of the algorithm can be done in $O(n)$ time. Each
iteration takes time proportional to the decrease in the quantity
$$2K\cdot |V(H)|+K\cdot\sum_{e\in E(H)}\lambda (e)+|L|+4\cdot\sum_{v\in
V(H)}C(v),$$ where $K$ is a sufficiently large constant. Thus the running-time
of the algorithm is $O(n)$.

\baselineskip 11pt
\vfill
\noindent
This material is based upon work supported by the National Science
Foundation under Grant No.~DMS-9970514.
Any opinions, findings, and conclusions or recommendations expressed
in this material are those of the authors and do not necessarily
reflect the views of the National Science Foundation.

\end{document}